\documentclass[a4paper, oneside, 11pt, reqno, final]{amsart}
\usepackage{simonfoldvik}

\input{Tocfix.tex}
\hypersetup{pdfauthor={Simon Foldvik}}
\hypersetup{pdftitle={Causal-Temporal Event Graphs: A Formal Model for Recursive Agent Execution Traces}}

\title[Causal-Temporal Event Graphs for Agent Execution Traces]{Causal-Temporal Event Graphs: \\ A Formal Model for Recursive Agent Execution Traces}

\begin{document}
    \maketitle
    \begin{center}
        \vspace{-2ex}
        \textsc{Simon Foldvik}\footnote{Independent researcher: \texttt{research@simonfoldvik.no}} \\[1ex]
        \textsc{April 19, 2026} \\
        \vspace{3ex}
    \end{center}

    \begin{quotation}
    \small
    \textbf{Abstract.}
    We introduce causal-temporal event graphs (CTEGs) as a formal model for fully resolved
    recursive agent execution records under single-parenthood causal semantics. We formalise direct
    event emissions and recursive subagent invocations as extension procedures on generic typed
    temporal graphs and show that the recursive closure \(\E_\infty\) of the induced maximal
    dynamics starting from single causal roots consists entirely of finite sequences of CTEGs. A
    CTEG is a rooted arborescence whose nodes carry timestamps and event types, subject to the
    constraint that timestamps be strictly increasing along causal paths. We realise \(\E_\infty\)
    as the increasing union of a recursive hierarchy \(\E_0 \subseteq \E_1 \subseteq \cdots\) of
    agent execution levels parametrised by recursion depth, which is recognised as the ascending
    Kleene chain of a monotone operator \(\varphi\) admitting \(\E_\infty\) as its least fixed
    point. Although the introduction of the full hierarchy is natural, stabilisation occurs already
    at \(\E_1\) if one insists that the internal construction of a subagent execution trace be a
    delegated and opaque computational unit. The CTEG formalism supports compositional construction
    of globally well-formed execution traces from local agent behaviour without centralised
    coordination, preserves well-formedness under partial execution failure, and admits a natural
    relational database encoding. The arborescent structure of CTEGs is further compatible with
    cryptographic Merkle tree commitments for tamper-evident session verification.
\end{quotation}
 
    {
        \setlength{\parskip}{2pt}
        \tableofcontents
    }
    \section{Introduction}
    Recent years have seen a surge in the use of \emph{agentic workflows}, wherein large language
    models (LLMs) play a central role in guiding artificial intelligence systems into solving
    generic, open-ended problems from natural language inputs. These systems are by their nature
    probabilistic, hence possess an inherent element of nondeterminism. This poses a challenge when
    such workflows are deployed to real-world decision-making systems in fields where regulatory
    frameworks impose requirements on explainability, auditability, and compliance with business
    and legal requirements.

    A common response is to capture the sequence \((e_k)\) of events produced by an agentic system
    (its execution log), enabling the reconstruction of full agent session chronologies for
    downstream inspection and replay. This is suitable for studying the \emph{temporal}
    relationships in agentic executions, but we demonstrate by means of example in
    \cref{Remark: Impossibility of causal reconstruction} that plain linear traces alone, without
    explicitly encoded causal parent structure, are not well-suited to also capture the underlying
    \emph{causal} relationships between events in agentic workflows. A tool output is the result
    of a tool invocation, and it matters not only in which order they come, but also which tools
    caused which outcomes. While agent execution tracing is an increasingly active engineering
    practice and area of research (see for instance \cite{PROV25, Szp26} and the references
    therein), formal treatments of recursive agent causal structure appear to remain limited.

    We introduce \emph{causal-temporal event graphs} (CTEGs) as a formal model for fully resolved
    recursive agent execution records under strict single-parenthood causal semantics. A CTEG is a
    rooted arborescence with typed nodes and timestamps, subject to the constraint that timestamps
    be strictly increasing along causal paths (\cref{Def: CTEG}). The formalism is motivated by
    distributed agentic production systems, wherein agents first declare their event type
    hierarchies and then asynchronously emit typed events in response to previous execution states,
    forming potentially independent cascading causal chains. The present version focuses on the
    formal model and its basic closure properties, with broader positioning relative to adjacent
    tracing frameworks deferred to future revisions.

    \begin{figure}[b]
        \includegraphics[width=\textwidth]{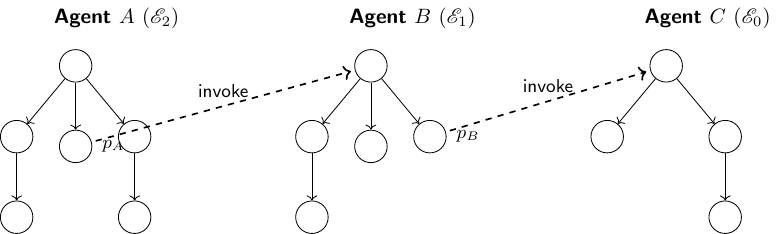}
        \caption{Recursive \(\E_2\)-agent execution. Timestamps and types are omitted for clarity.}
        \label{Fig: Recursive E2-execution trace}
    \end{figure}

    The computational model we have in mind is that of a parent agent invoking any number of worker
    agents (the subagents), which in turn may invoke further subagents as part of their own
    executions. Subagents are responsible for the construction of their own internal execution
    traces, passing the finished subtraces up to the invoking parent agents as completed units. The
    parents then atomically graft the subagent traces into their own execution records at the
    subagent invocation nodes. This creates an opaque subagent invocation interface and a natural
    separation of concerns: Each subtrace construction remains small and records precisely the
    causal relationships the corresponding agent is in a position to assert. The procedure is
    illustrated in \cref{Fig: Recursive E2-execution trace}, with the resulting global execution
    trace depicted in \cref{Fig: Global E2-execution trace}.

    The main contribution of this paper is the identification and formalisation of CTEGs as natural
    objects for capturing fully resolved global execution traces in recursive agentic systems under
    single-parenthood causal semantics, together with a local-to-global compositional model of
    their construction. We formalise two local dynamical operations (direct event emissions and
    recursive subagent invocations) and show that the smallest set \(\E_\infty\) of execution
    sequences closed under these operations starting from single causal roots (the recursive
    closure) consists entirely of finite sequences of CTEGs
    (\cref{Theorem: CTEG characterisation}). This gives a formal account of how globally
    well-formed causal structure can be assembled in recursive agentic systems from delegated
    subagent trace constructions without centralised coordination, and this remains true even in
    the face of partial agent failure. The resulting arborescent traces admit natural relational
    database encodings and support tamper-evident cryptographic commitments via Merkle tree
    computations.

    The proof proceeds by introducing a recursive hierarchy
    \begin{equation}
        \E_0 \subseteq \E_1 \subseteq \cdots
    \end{equation}
    of agent execution levels parametrised by recursion depth, whose union
    \(\smash{\E_\infty = \bigcup \E_d}\) is shown to be the least fixed point of a monotone
    operator \(\varphi\) on the power set of all finite sequences of typed temporal graphs. Level
    \(\E_0\) contains only direct event emissions, while \(\E_{d+1}\) is constructed from level
    \(\E_d\) by additionally allowing recursive \(\E_d\)-subagent invocations. The causal links are
    built and maintained at this local level through child node pointers to either the newly
    emitted events or the causal root of the recursive \(\E_d\)-subagent invocation. Even though
    the introduction of the full \(\E_d\)-hierarchy is natural, stabilisation occurs already at
    \(\E_1\). This is a consequence of the opacity of the subagent invocation interface, wherein
    parent agents receive only the completed subagent execution records and not their construction
    histories. Any CTEG admits an \(\E_0\)-construction by topological sorting, hence any
    \(\E_d\)-invocation with \(d \ge 1\) might as well have been an \(\E_0\)-invocation from the
    parent's point of view.

    The single-parenthood semantics is partly motivated by the following. Suppose a parent agent
    invokes a subagent, after whose completion the parent decides to emit a continuation event.
    Although events internal to the subagent execution may have influenced the parent's decision to
    produce the continuation, or even informed the contents of the continuation event itself, the
    parent agent has no principled basis (short of inspecting the subagent's internal execution)
    for determining which internal subagent events should be causally linked to the continuation.
    To remedy this tension one could either relax the insistence that global execution traces be
    fully resolved in place (allowing subagent invocations as atomic execution graph nodes), or
    parent agents could break epistemic locality and inspect subagents' internal affairs. We
    discuss these limitations and alternatives further in \cref{Section: Conclusion}.

    \begin{wrapfigure}{r}{0.5\textwidth}
        \centering
        \includegraphics[width=0.48\textwidth]{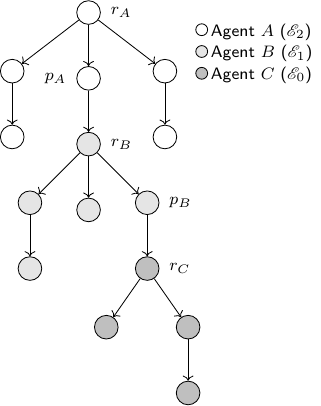}
        \caption{Global \(\E_2\)-execution trace for the agentic computation outlined in
        \cref{Fig: Recursive E2-execution trace}. Timestamps and types are omitted for clarity.}
        \label{Fig: Global E2-execution trace}
    \end{wrapfigure}
    
    Our theory thus addresses the additional challenge of building and maintaining
    \emph{well-formed} global causal structure in recursive agentic systems and shows that it
    suffices to solve this locally: Each agent manages its own (small) execution trace internally,
    subagent traces are grafted in at the parent agent's execution root upon completion (or
    failure), and the resulting global trace is guaranteed to be a well-formed CTEG. The question
    of what constitutes \emph{correct} (as opposed to well-formed) causal structure must be
    determined against choices of local parenthood semantics, which are orthogonal to the CTEG
    formalism in its current form. Instead of asserting which local semantics are most appropriate,
    CTEGs encode those local assignments once made and are sufficiently general to accommodate
    different semantic choices. Discussing which local semantic choices \emph{should} be made is
    therefore beyond the scope of the current paper.

    \subsection{Outline}
        The paper is organised as follows. In \cref{Section: CTEGs} we introduce causal-temporal
        event graphs, the criterion under which two CTEGs may be composed to preserve the CTEG
        structure, along with their relational encoding and cryptographic commitments via Merkle
        hashes. In \cref{Section: Agent Execution Hierarchy} we formalise the recursive
        \(\E_d\)-hierarchy model of cascading agent executions based on direct event emissions and
        subagent invocations, characterise their union \(\E_\infty\) both as the least fixed point
        of a monotone operator and the recursive closure of the maximal dynamics, and prove that
        all agent execution sequences arising in this way are finite sequences of CTEGs. We
        conclude in \cref{Section: Conclusion} with a brief discussion of limitations and natural
        extensions of the current work.

    \subsection{Notation}
        We write `iff' in place of `if and only if', \(0 \in \N\) is a natural number,
        \(\nodes(G)\) denotes the set \(N\) of nodes of a directed graph \(G = (N, E)\), and
        \(\edges(G) \subseteq N \times N\) its set \(E\) of edges. If \(f \colon A \to C\) and
        \(g \colon B \to C\) are maps with disjoint domains, then \(f \sqcup g\) is the unique map
        \(A \cup B \to C\) which restricts to \(f, g\) on \(A, B\), respectively. \(\pwset(X)\)
        denotes the power set of \(X\).

    \section{Causal-Temporal Event Graphs}
    \label{Section: CTEGs}
    In this section we present the static view of causal-temporal event graphs (CTEGs), a formal
    model for encoding causal and temporal relationships in recursive agent execution traces along
    with their event types under single-parenthood causal semantics. We state a fundamental lemma
    on when two CTEGs may be composed to preserve the CTEG structure, and show that causal
    structure may not be faithfully recovered from temporal structure alone. We observe that CTEGs
    admit natural relational database encodings and support tamper-evident cryptographic
    commitments via Merkle hashes. This lays the foundation for the formal treatment of the agent
    execution graph dynamics presented in \cref{Section: Agent Execution Hierarchy}.

    \subsection{Causal, Temporal, and Type Structures}
        By a \emph{causal graph} we mean a finite, non-empty directed acyclic graph \(G = (N, E)\)
        (where \(\nodes(G) \coloneqq N\) is the node set and
        \(\edges(G) \coloneqq E \subseteq N \times N\) the edge set), with a distinguished root
        node \(r \in \nodes(G)\) (called the \emph{causal root} and henceforth denoted
        \(r_G \coloneqq r\)) such that \(G\) is an arborescence with root \(r\):
        \begin{enumerate}
            \item
                All nodes \(n \in \nodes(G) \setminus \set{ r }\) are reachable from \(r\).

            \item
                Every \(n \in \nodes(G) \setminus \set{ r }\) has exactly one incoming edge.

            \item
                There are no edges into \(r\).
        \end{enumerate}

        \begin{remark}
            In a causal graph \(G\), the causal path from the root \(r\) to any
            \(n \in \nodes(G) \setminus \set{ r }\) is in fact unique. If not, choose \(n \ne r\)
            admitting two distinct directed walks \(w, w'\) from \(r\) minimising
            \(\l(w) + \l(w')\) (sum of lengths) by well-ordering. Since \(n\) has a unique incoming
            edge, both walks must pass through the same immediate predecessor, yielding a smaller
            counterexample. Contradiction.
        \end{remark}

        By a \emph{temporal structure} on a causal graph \(G\) we mean an assignment
        \(t \colon \nodes(G) \to \R\) of timestamps to the nodes. It is \emph{causally compatible}
        with \(G\) iff it is strictly increasing along causal paths: If \(m, n \in \nodes(G)\) and
        \(n\) is reachable from \(m\), then \(t(m) < t(n)\). It is enough to verify this criterion
        along directed \emph{edges} \((m, n) \in \edges(G)\) by transitivity.

        \begin{remark}[Cause-effect]
            Hence for a causally compatible temporal structure, a cause must strictly precede its
            effects, but sibling nodes sharing the same causal parent are allowed to be
            simultaneous.
        \end{remark}

        \begin{remark}[Tie-breaking]
            Our temporal structures assume real-valued timestamps with strict ordering along causal
            paths. In practice, finite clock resolution may require tie-breaking mechanisms when
            two causally related events receive identical timestamps from the system clock.
        \end{remark}

        Let \(T\) be a finite, non-empty set of \emph{types}. By a \emph{type map} on a causal
        graph \(G\) we mean an assignment \(\tau \colon \nodes(G) \to T\) of types to its nodes.
        Combining all of the above, we obtain the central structure of this paper.

        \begin{definition}[Causal-temporal event graphs]
            \label{Def: CTEG}
            A \emph{causal-temporal event graph} (CTEG) over a type set \(T\) is a causal graph
            \(G\) equipped with a causally compatible temporal structure
            \(t \colon \nodes(G) \to \R\) and a type map \(\tau \colon \nodes(G) \to T\). The full
            CTEG is denoted \((G, r, t, \tau)\), where \(r\) is the unique causal root.
        \end{definition}
        \begin{remark}[Flat type structures]
            The type set \(T\) is flat by design: A CTEG records event types but imposes no
            relations between them. Richer type structures, such as partially ordered type
            hierarchies, are treated as application-layer concerns rather than part of the
            causal-temporal formalism.
        \end{remark}
        \begin{remark}[Cryptographic session commitments]
            The arborescent structure of CTEG execution graphs admits tamper-evident cryptographic
            commitments of agent session histories, partial or complete, via Merkle tree
            computations. The Merkle hash \cite{Mer90} of a CTEG may thus serve as a
            cryptographically verifiable receipt of an agentic computation.
        \end{remark}
        \begin{remark}[Relational database encoding]
            CTEG session traces admit natural relational representations in standard database
            formats. This is robust also in the face of agent failure, as even partial execution
            traces of agentic computations following the assumptions of this paper are guaranteed
            to be CTEGs.

            To represent a CTEG \((G, r, t, \tau)\) in relational database format, first register
            the session to receive a globally unique session identifier. Each node of \(G\) is then
            written into an append-only node table, receiving a unique node identifier, a timestamp
            derived from the temporal structure \(t\), a reference to the session identifier, the
            node type derived from the type map \(\tau\), and an opaque payload column for the node
            data. Each non-root node additionally stores a reference to its unique causal parent
            node, which exists and is unique by the arborescence property. The root node has no
            parent reference.

            With agentic CTEG session persistence as outlined above, full causal and temporal
            session reconstruction reduces to a single query fetching all nodes with a given
            session identifier followed by pointer resolution.
        \end{remark}

    \subsection{Grafting}
        Given two causal graphs \((G_1, r_1)\) and \((G_2, r_2)\) with disjoint node sets, there is
        a natural \emph{grafting} operation attaching the root of \(G_2\) to a select node
        \(p \in \nodes(G_1)\) (see \cref{Fig: Grafting operation}). It is denoted
        \(G_1 \oplus_p (G_2, r_2)\) and defined as follows.
        \begin{figure}[ht]
            \includegraphics[scale=1]{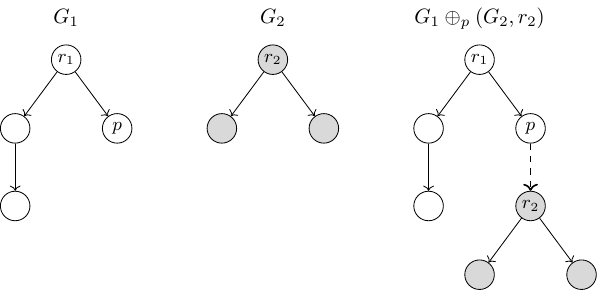}
            \caption{Illustrating the graft of a causal graph \(G_2\) into \(G_1\).}
            \label{Fig: Grafting operation}
        \end{figure}

        \begin{definition}[Graft]
            \label{Def: Grafting operation}
            The \emph{graft} \(G_1 \oplus_p (G_2, r_2)\) of two causal graphs \(G_1\) and \(G_2\)
            with disjoint node sets at \(p \in \nodes(G_1)\) is the directed graph \(G\) defined by
            \begin{equation}
                \begin{aligned}
                    \nodes(G) &\coloneqq \nodes(G_1) \cup \nodes(G_2), \\
                    \edges(G) &\coloneqq \edges(G_1) \cup \edges(G_2) \cup \set[\big]{ (p, r_2) }.
                \end{aligned}
            \end{equation}
        \end{definition}
        \begin{remark}
            The graft \(G\) defined in \cref{Def: Grafting operation} is thus a directed graph
            consisting of the same nodes and edges as its constituent subgraphs, but with a new
            edge joining \(p \in \nodes(G_1)\) to \(r_2 \in \nodes(G_2)\) adjoined. It is again a
            causal graph with causal root \(r_1\). This follows since the grafting operation only
            realises the following structural changes:
            \begin{itemize}
                \item
                    The in-degree changes from zero to one at the causal root \(r_2\) of \(G_2\).

                \item
                    All other nodes retain their in-degrees.

                \item
                    Reachability of all nodes from \(r_1\) holds by transitivity through \(p\).

                \item
                    No cycles are created, since any cycle would be pushed into one of the acyclic
                    subgraphs \(G_1\) or \(G_2\) since it could not have traversed the new edge
                    \((p, r_2)\) whilst starting and ending at the same node.
            \end{itemize}
        \end{remark}

        The following lemma answers the question of when the graft of two CTEGs preserves the CTEG
        structure. This is the case iff the adjoined grafting edge respects the causal-temporal
        compatibility criterion.

        \begin{lemma}[CTEG composition lemma]
            \label{Lemma: CTEG composition}
            Let \((G_i, r_i, t_i, \tau_i)\) for \(i = 1, 2\) be two CTEGs with type sets \(T_i\)
            and disjoint node sets. Let
            \begin{equation}
                G \coloneqq G_1 \oplus_p (G_2, r_2)
            \end{equation}
            be their graft at \(p \in \nodes(G_1)\). Then \(G\) equipped with the temporal
            structure \(t_1 \sqcup t_2\), type map \(\tau_1 \sqcup \tau_2\), and causal root
            \(r_1\) is a CTEG over \(T_1 \cup T_2\) iff
            \begin{equation}
                t_1(p) < t_2(r_2).
            \end{equation}
            We call this the \emph{causal-temporal compatibility criterion} for CTEG grafting.
        \end{lemma}
        \begin{remark}
            Recall that \(t_1 \sqcup t_2 \colon \nodes(G) \to \R\) is the unique map which
            restricts to \(t_i\) on \(\nodes(G_i)\) for \(i = 1, 2\), and similarly for
            \(\tau_1 \sqcup \tau_2\).
        \end{remark}
        \begin{proof}
            Immediate.
        \end{proof}
        \begin{remark}[Type contracts]
            This lemma mirrors the setup in production systems wherein agent \(A_1\) first declares
            its type set \(T_1\) to communicate its possible event types to type checkers, and
            similarly agent \(A_2\) declares \(T_2\). The execution of \(A_1\) alone would produce
            a CTEG over \(T_1\), whereas the execution of \(A_1\) invoking \(A_2\) as a subagent
            would produce a CTEG over \(T_1 \cup T_2\). This avoids the need for a centralised type
            hierarchy and makes the streaming contract between agents and their invokers explicit.
        \end{remark}
        \begin{remark}
            \label{Remark: Impossibility of causal reconstruction}
            We conclude this section by observing that causal structure may not be faithfully
            reconstructed from temporal structure alone. A counterexample will suffice. By a
            \emph{temporal projection} of a causal graph \(G\) with respect to a causally
            compatible temporal structure \(t \colon \nodes(G) \to \R\), we mean a bijective
            enumeration \((n_i)_{i=0}^m\) of its node set \(\nodes(G)\) such that
            \begin{equation}
                t(n_0) \le \cdots \le t(n_m).
            \end{equation}
            It is now easy to see that \((r, a, b, c)\) is the temporal projection of both \(G\)
            and \(G'\) in \cref{Fig: Causal reconstruction counterexample}, even though these have
            different causal structures.
            \begin{figure}[ht]
                \includegraphics[scale=1]{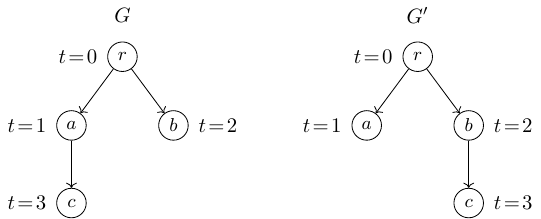}
                \caption{Two causal structures with the same temporal projection.}
                \label{Fig: Causal reconstruction counterexample}
            \end{figure}

            This precludes the hope of any strategy reconstructing a causal structure from the
            temporal structure alone. In particular, any implementation attempting to reconstruct a
            CTEG by observing the temporal sequence of events emitted by an agent (for instance, by
            externally decorating an asynchronous generator and listening to the temporal sequence
            of emitted events) will lose causal information unless it is also willing to inspect
            internal agent states (or those of its emitted events) in an attempt to infer the
            missing causal links. The causal structure is therefore best constructed and maintained
            internally by agents during execution, rather than inferred after the fact.
        \end{remark}

    \section{Recursive Execution and Stabilisation}
    \label{Section: Agent Execution Hierarchy}
    In this section we demonstrate how CTEGs arise naturally under the iterated local dynamics of
    direct event emissions and recursive subagent invocations starting from single causal roots.
    More precisely, we prove that the smallest set \(\E_\infty\) closed under these operations
    starting from single causal roots (its recursive closure) consists entirely of finite sequences
    of CTEGs.

    For set-theoretic reasons we fix an infinite set \(\A\) of \emph{actions} throughout this
    section to constrain the size of sets constructed in the following procedure. We also fix a
    finite type set \(T\). These will be left implicit in the notation.

    We start by introducing the ambient space \(\G\) of all finite sequences \((G_k)_{k=0}^n\) of
    \emph{typed temporal graphs}. These are directed graphs with types and timestamps, bearing no
    assumptions of being rooted or of compatibility between the causal and temporal structures. The
    idea is that a dynamical model of agent execution sequences recording types along with causal
    and temporal relationships must necessarily live in this space. We then single out
    \(\E_\infty \subset \G\) as the smallest set of such sequences
    \begin{equation}
        G_0 \gle \cdots \gle G_n
    \end{equation}
    starting from a single causal root \(G_0\) and evolving by either direct emissions or
    self-invocations at each step. \cref{Theorem: CTEG characterisation} proves that \(\E_\infty\)
    then consists entirely of finite sequences of CTEGs.

    We further introduce a recursive hierarchy \(\E_0 \subseteq \E_1 \subseteq \cdots\) of subagent
    execution levels parametrised by recursion depth and show that \(\E_\infty\) may be realised as
    the union \(\bigcup \E_d\) over all sublevels. In this connection we identify a monotone
    operator \(\varphi\) on the power set \(\pwset(\G)\) of \(\G\) exhibiting the \(\E_d\) as
    iterates \(\smash{\E_d = \varphi^{(d+1)}(\emptyset)}\) of the empty set and admitting
    \(\E_\infty\) as its smallest fixed point. Thus \(\varphi(\E_\infty) = \E_\infty\) and
    \(\E_\infty \subseteq \F\) whenever \(\F \subseteq \G\) and \(\varphi(\F) = \F\). This realises
    the \((\E_d)_{d \in \N}\) execution hierarchy as a simple recursion on \(\pwset(\G)\) under
    \(\varphi\) starting at \(\emptyset\) yielding \(\E_\infty\) as the canonical limit.

    The operator \(\varphi\) encodes the \emph{logical} progression between successive recursion
    levels: At each step, an agent either emits new events directly, or invokes a subagent whose
    internally completed execution sequence \((H_i)\) yields a final element \(H_m\) to be grafted
    in at the invocation root. The step-by-step temporal dynamics of the subagent's construction
    are \emph{not} reflected in the parent's execution sequence \((G_k)\): From the parent's
    perspective, the subagent's execution graph is received as a completed unit upon subagent
    termination and grafted atomically to enable separation of concerns.

    Even though the introduction of the full \(\E_d\)-hierarchy is natural, it turns out that
    stabilisation occurs already at level \(\E_1\), in the sense that \(\E_d = \E_1\) for all
    \(d \ge 1\) (but \(\E_0 \ne \E_1\)). The underlying reason is that the subagent invocation
    interface is opaque to construction history: It inspects only the final element \(H_m\) of a
    subagent trace sequence \((H_i)\). Any CTEG admits an \(\E_0\)-construction by topological
    sorting, hence every \(\E_d\)-invocation step can be replicated as an \(\E_0\)-invocation step
    by substituting an \(\E_0\)-sequence with the same final element.

    The construction of the \(\E_d\)-hierarchy by means of simple recursion of \(\varphi\) from the
    empty set can be seen in relation to Kleene's theorem in the complete lattice
    \(\smash{(\pwset(\G), \subseteq)}\). A compactness argument gives the continuity of
    \(\varphi\), and \((\E_d)_{ d \in \N }\) sits in the ascending Kleene chain
    \begin{equation}
        \emptyset
            \subseteq \varphi(\emptyset)
            \subseteq \varphi\paren[\big]{ \varphi(\emptyset) }
            \subseteq \cdots
    \end{equation}
    converging to the least fixed point \(\E_\infty = \bigcup \E_d\) of \(\varphi\) in the limit.

    \subsection{Local Dynamics}
        \label{Section: Local Dynamics}
        We write \(G \le G'\) when \(G, G'\) are directed graphs and \(G'\) extends \(G\) in the
        sense that \(\nodes(G) \subseteq \nodes(G')\) and \(\edges(G) \subseteq \edges(G')\). All
        nodes and edges of \(G\) are thus also in \(G'\) when \(G\) is a subgraph of \(G'\), but
        not necessarily the other way around. A node map \(f \colon \nodes(G') \to X\) on \(G'\)
        thus restricts to a node map \(\nodes(G) \to X\) on \(G\) when \(G \le G'\), an object we
        denote by \(f\restr_G\).

        By a \emph{typed temporal graph} (over \(T\), with nodes in \(\A\)) we mean a triple
        \((G, t, \tau)\), where \(G\) is a directed graph over \(\A\) with temporal structure
        \(t \colon \nodes(G) \to \R\) and type map \(\tau \colon \nodes(G) \to T\). We make no
        assumptions that these be rooted and will for notational simplicity write \(G\) in place of
        \((G, t, \tau)\) when \(t\) and \(\tau\) are clear from the context. The subgraph relation
        extends to typed temporal graphs to include compatibility conditions on the temporal and
        type structures under restrictions: We declare
        \begin{equation}
            (G, t, \tau) \gle (G', t', \tau')
        \end{equation}
        iff \(G \le G'\) and the compatibility conditions \(t'\restr_G = t\) and
        \(\tau'\restr_G = \tau\) hold.

        \begin{definition}
            We let \(\G\) denote the set of all finite sequences \((G_k)_{k=0}^n\) of typed
            temporal graphs \((G_k, t_k, \tau_k)\) over \(\A\) with type set \(T\).
        \end{definition}

        The next definition models the direct emission of new nodes by an agent in its current
        execution state (see \cref{Fig: Direct emissions}). This defines the first of the two local
        dynamical assumptions of our paper. We deliberately choose not to model simultaneous
        emissions from distinct roots, since this would complicate the arguments that follow
        without enlarging the class of attainable execution states.

        \begin{definition}[Direct emissions]
            \label{Def: Emission semantics}
            We say \(G \gle G'\) \emph{by direct emissions} for typed temporal graphs \(G, G'\) iff
            there exists a finite, non-empty set \(A \subset \A\) of new actions disjoint from
            \(\nodes(G)\) and an \emph{emission root} \(p \in \nodes(G)\) such that
            \begin{equation}
                G' = G \oplus_p A
            \end{equation}
            with \(t(p) < t'(a)\) for all \(a \in A\).
            \begin{figure}
                \includegraphics[scale=1]{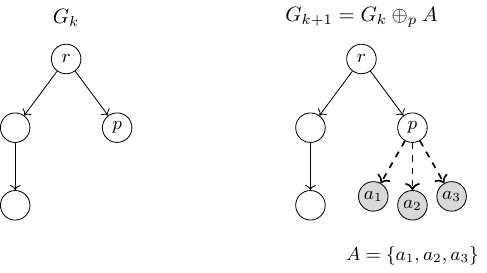}
                \caption{Evolution from state \(G_k\) to \(G_{k+1}\) by direct emissions.}
                \label{Fig: Direct emissions}
            \end{figure}
        \end{definition}
        \begin{remark}
            Of course, \(G \oplus_p A\) (where \(\nodes(G) \cap A = \emptyset\)) denotes the
            directed graph with node set \(\nodes(G) \cup A\) and edge set
            \(\edges(G) \cup \set{ (p,a) : a \in A }\).
        \end{remark}

        The second dynamical mode of our paper is the recursive invocation of a subagent, awaiting
        its execution graph, and grafting the resulting structure in at the parent agent's
        invocation root (see \cref{Fig: Subagent invocation}). So as to not make any a priori
        assumptions on the structure of these recursive invocations, we state our definition with
        respect to general finite sequences \((G_k)_{k=0}^n\) of typed temporal graphs
        \((G_k, t_k, \tau_k)\) and recall that \(\G\) denotes the set of all such.

        \begin{figure}
            \includegraphics[scale=1]{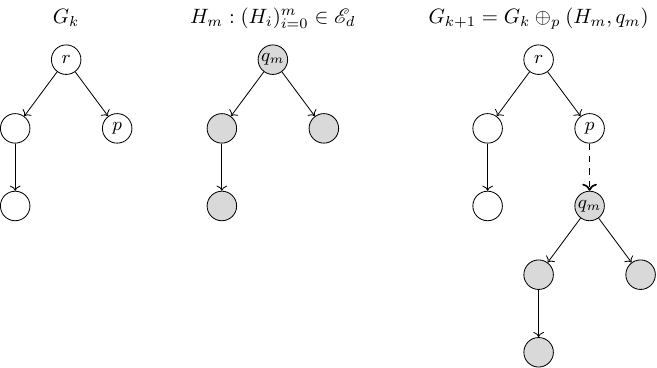}
            \caption{Evolution from state \(G_k\) to \(G_{k+1}\) by \(\E_d\)-subagent invocation.}
            \label{Fig: Subagent invocation}
        \end{figure}

        \begin{definition}[Subagent invocations]
            \label{Def: Subagent invocation semantics}
            Given a set of sequences \(\E \subseteq \G\) and typed temporal graphs \(G, G'\), we
            say \(G \gle G'\) \emph{by \(\E\)-invocation} iff there exists an \emph{invocation
            root} \(p \in \nodes(G)\) along with some \((H_i)_{i=0}^m \in \E\) with node sets
            disjoint from \(G\), and some \(q_m \in \nodes(H_m)\) of in-degree zero, such that
            \begin{equation}
                G' = G \oplus_{p} (H_m, q_m),
            \end{equation}
            where \(t' = t \sqcup t_m\), \(\tau' = \tau \sqcup \tau_m\), and \(t(p) < t_m(q_m)\)
            for compatibility.
        \end{definition}
        \begin{remark}
            Here \(t_m \colon \nodes(H_m) \to \R\) and \(\tau_m \colon \nodes(H_m) \to T\) denote
            the temporal and type structures of \(H_m\), respectively. Note also that when \(H_m\)
            is a CTEG, there is exactly one choice of \(q_m\), namely its causal root.
        \end{remark}

    \subsection{The Recursive Execution Hierarchy}
        We now introduce the \emph{agent execution hierarchy} \((\E_d)_{ d \in \N }\). A directed
        graph \(G\) will be said to be \emph{\(\A\)-trivial} iff there exists \(r \in \A\) such
        that \(\nodes(G) = \set{ r }\) and \(\edges(G) = \emptyset\). In other words, \(G\) is
        \(\A\)-trivial iff it has a single causal root in \(\A\) and no edges.
        \begin{enumerate}
            \item
                Let \(\E_0 \subset \G\) be the set of all finite sequences \((G_k)_{k=0}^n \in \G\)
                of typed temporal graphs such that \(G_0\) is \(\A\)-trivial and
                \begin{equation}
                    G_0 \gle \cdots \gle G_n
                \end{equation}
                by direct emissions at each step.

            \item
                Supposing recursively that \(\E_d \subset \G\) has been defined for some
                \(d \in \N\), define \(\E_{d+1}\) to be the set of all finite sequences
                \((G_k)_{k=0}^n \in \G\) such that \(G_0\) is \(\A\)-trivial and
                \begin{equation}
                    G_0 \gle \cdots \gle G_n
                \end{equation}
                by either direct emissions or \(\E_d\)-invocations at each step.
        \end{enumerate}
        
        \begin{remark}
            The recursion depth indexes an abstraction boundary: An \(\E_d\)-computation is a
            self-contained unit of causal structure, composable into any \(\E_{d+1}\)-computation
            meeting the causal-temporal compatibility criterion without the parent inspecting the
            subagent's internal execution history.
        \end{remark}

        This may be viewed as a simple recursion on the power set \(\pwset(\G)\) of \(\G\). Let
        \(\varphi \colon \pwset(\G) \to \pwset(\G)\) take a set \(\E \subseteq \G\) of (candidate)
        execution sequences to the set \(\varphi(\E) \subseteq \G\) of all \((G_k) \in \G\) with
        \(G_0\) \(\A\)-trivial and \(G_0 \gle \cdots \gle G_n\) by either direct emissions or
        \(\E\)-invocations at each step. Then
        \begin{equation}
            \label{Eq: Ed as simple recursion}
            \E_0 = \varphi(\emptyset) \qand \E_{d+1} = \varphi(\E_d) \quad (d \in \N).
        \end{equation}
        Hence \(\E_d = \varphi^{(d+1)}(\emptyset)\) (compositions) for all \(d \in \N\).

        \begin{lemma}
            \(\varphi\) is monotone: If \(\E \subseteq \E'\) in \(\pwset(\G)\), then
            \(\varphi(\E) \subseteq \varphi(\E')\).
        \end{lemma}
        \begin{proof}
            Immediate.
        \end{proof}

        Clearly \(\E_0 \subseteq \E_1 \subseteq \cdots\) follows, and we put
        \begin{equation}
            \E_\infty \coloneqq \bigcup_{ d \in \N } \E_d.
        \end{equation}

        \begin{definition}
            \(\E_\infty\) is the set of \emph{agent execution sequences}.
        \end{definition}
        \begin{remark}
            See \cref{Fig: Recursive E2-execution trace,Fig: Global E2-execution trace} for a
            concrete example of an agent-subagent execution sequence to help unpack the formalism.
        \end{remark}

        We now give an intrinsic characterisation of \(\E_\infty\) without reference to its
        construction through the \(\E_d\)-hierarchy.

        \begin{lemma}[Recursive closure]
            \label{Lemma: Fixed-point theorem}
            \(\E_\infty\) is the smallest fixed point of \(\varphi\):
            \begin{equation}
                \varphi(\E_\infty) = \E_\infty,
            \end{equation}
            and if \(\F \subseteq \G\) is such that \(\varphi(\F) = \F\), then
            \(\E_\infty \subseteq \F\).
        \end{lemma}
        \begin{remark}
            The same holds if merely \(\varphi(\F) \subseteq \F\).
        \end{remark}
        \begin{proof}
            For each \(d \in \N\) it holds that
            \(\E_d \subseteq \E_{d+1} = \varphi(\E_d) \subseteq \varphi(\E_\infty)\), hence
            \(\E_\infty = \bigcup \E_d \subseteq \varphi(\E_\infty)\), and this takes care of the
            first inclusion. To prove \(\varphi(\E_\infty) \subseteq \E_\infty\), note that any
            \((G_k) \in \varphi(\E_\infty)\) evolves from \(G_0\) \(\A\)-trivial by either direct
            emissions or \(\E_\infty\)-invocations at each step. Since there are finitely many
            steps, there exists \(d \in \N\) (compactness) such that
            \((G_k) \in \varphi(\E_d) = \E_{d+1} \subseteq \E_\infty\). This proves that
            \(\E_\infty\) is a fixed point of \(\varphi\).

            To prove it is the smallest, assume \(\F \subseteq \G\) satisfies \(\varphi(\F) = \F\),
            and observe first that \(\E_0 = \varphi(\emptyset) \subseteq \varphi(\F) = \F\).
            Furthermore, if \(\E_d \subseteq \F\) for some \(d \in \N\), then also
            \(\E_{d+1} = \varphi(\E_d) \subseteq \varphi(\F) = \F\), hence
            \(\E_\infty = \bigcup \E_d \subseteq \F\) by induction.
        \end{proof}
        \begin{remark}
            \label{Remark: Fixed point literature}
            Similar fixed point arguments are standard in domain theory. The reader can find in
            Abramsky \& Jung \cite[Theorem~2.1.19]{AJ94} the more general statement that every
            Scott-continuous (monotone) operator \(f\) on a pointed DCPO (partially ordered set in
            which every directed subset has a supremum) \((P, \bot)\) admits a least fixed point
            realised as the supremum of the ascending \(f\)-chain \((f^{(k)}(\bot))_{ k \in \N }\)
            starting at the bottom element \(\bot \in P\). Tarski \cite{Tar55} gives fixed points
            for monotone operators on complete lattices without continuity assumptions, at the
            expense of less constructive realisations, and contains interesting pointers to similar
            results in the literature. See also Kleene \cite{Kle52}.
        \end{remark}

        The compactness argument in the proof of \cref{Lemma: Fixed-point theorem} in fact
        generalises to establish continuity, recovering the connections outlined in
        \cref{Remark: Fixed point literature}.

        \begin{corollary}
            \(\varphi\) is Scott-continuous on the power set lattice \(\pwset(\G)\).
        \end{corollary}
        \begin{proof}
            To see this, let
            \((\F_i)\) be a directed family in \(\pwset(\G)\), and we need to prove
            \(\smash{\varphi(\bigcup \F_i) = \bigcup \varphi(\F_i)}\). The inclusion
            \(\smash{\varphi(\F_j) \subseteq \varphi(\bigcup \F_i)}\) holds for all \(j\) since
            \(\varphi\) is monotone, and from this follows the first inclusion
            \(\smash{\bigcup \varphi(\F_j) \subseteq \varphi(\bigcup \F_i)}\). To prove the reverse
            inclusion, consider \((G_k) \in \varphi(\bigcup \F_i)\), so that \(G_0\) is
            \(\A\)-trivial and we pick for each \(k\) for which \(G_{k+1}\) is obtained from
            \(G_k\) by \(\smash{\bigcup \F_i}\)-invocation an index \(i_k\) such that the
            invocation sequence lies in \(\F_{i_k}\). Since \((\F_i)\) is directed and there are
            finitely many steps, there is \(l\) such that \(\F_{i_k} \subseteq \F_l\) for all
            \(k\), and from this follows
            \(\smash{(G_k) \in \varphi(\F_l) \subseteq \bigcup \varphi(\F_j)}\).
        \end{proof}

        The following theorem characterises the objects contained in \(\E_\infty\).

        \begin{theorem}
            \label{Theorem: CTEG characterisation}
            \(\E_\infty\) is the smallest set of execution sequences (subsets of \(\G\)) closed
            under direct emissions and self-invocations starting at \(\A\)-trivial causal roots.
            All its elements are finite sequences of CTEGs over \(T\) with nodes in \(\A\),
            evolving according to the emission and invocation semantics of
            \cref{Section: Local Dynamics}.
        \end{theorem}
        \begin{proof}
            If all initial states of elements of \(\F \subseteq \G\) are \(\A\)-trivial and
            \(\varphi(\F) \subseteq \F\), it is proved as in \cref{Lemma: Fixed-point theorem} that
            in fact \(\E_\infty \subseteq \F\). Hence \(\E_\infty\) is the smallest subset of
            \(\G\) closed under direct emissions and self-invocations starting from single causal
            roots.

            That \(\E_\infty\) consists entirely of finite sequences of CTEGs is an induction on
            recursion depth in the decomposition \(\smash{\E_\infty = \bigcup \E_d}\). One checks
            it holds for \(\E_0\) since any \((G_k) \in \E_0\) starts at \(G_0\) \(\A\)-trivial
            (and is therefore a CTEG) and proceeds only by direct emissions which preserve the CTEG
            structure at each step. If \(d \in \N\) is such that the claim holds for \(\E_d\), one
            argues as for \(\E_0\) that the same must be true for \(\E_{d+1}\) as well. Indeed, if
            \((G_k) \in \E_{d+1}\), the only change in the argument is when \(G_{k+1}\) is obtained
            by grafting in the final element \(H_m\) of some subagent trace sequence
            \((H_i) \in \E_d\). In this case, \(H_m\) is a CTEG by the induction hypothesis (with
            \(q_m\) the causal root), hence the same is true of \(G_{k+1}\) by the inner induction
            on \(k\) and the CTEG composition lemma (\cref{Lemma: CTEG composition}).
        \end{proof}

        \begin{remark}[Robustness to partial failure]
            Since well-formedness is preserved at each agent execution step, any prefix of an agent
            execution sequence produces a valid CTEG. In particular, if a session terminates
            prematurely due to failure, the partial trace retains full causal and temporal
            integrity up to that point and therefore remains directly serialisable.
        \end{remark}

        As promised, we prove stabilisation of the \(\E_d\)-hierarchy already at \(\E_1\).

        \begin{theorem}
            One has \(\E_0 \ne \E_1 = \E_\infty\), hence \(\E_d = \E_1\) for all \(d \ge 1\).
        \end{theorem}
        \begin{proof}
            The inclusion \(\E_0 \subset \E_1\) is strict: An \(\E_0\)-sequence \((H_i)\) of two
            successive emission steps may produce a final CTEG \(H_m\) of height at least \(2\).
            Grafting \(H_m\) in a single \(\E_0\)-invocation step at an appropriate leaf node of
            the parent agent's execution graph may increase its height by more than one. Since this
            is not possible by direct emissions alone, it follows that \(\E_0 \ne \E_1\).

            To prove \(\E_1 = \E_\infty\), it suffices to show that \(\varphi(\E_1) = \E_1\) since
            \(\E_\infty\) is the smallest fixed point of \(\varphi\) and
            \(\E_1 \subseteq \E_\infty\). Let \((G_k) \in \varphi(\E_1)\). Each invocation step
            grafts the final element \(H_m\) of some \(\E_1\)-construction sequence \((H_i)\).
            Since \(H_m\) is a CTEG, hence a finite arborescence with causally compatible
            timestamps, a topological sort yields an \(\E_0\)-sequence with \(H_m\) as the final
            element. Hence every \(\E_1\)-invocation step is also an \(\E_0\)-invocation step, and
            so \((G_k) \in \varphi(\E_0) = \E_1\). This proves that
            \(\varphi(\E_1) \subseteq \E_1\), and the reverse inclusion follows from
            \eqref{Eq: Ed as simple recursion} onwards.
        \end{proof}

    \section{Conclusion}
    \label{Section: Conclusion}
    The present work has introduced causal-temporal event graphs (CTEGs) as a formal model for
    fully resolved recursive agent execution records under single-parenthood causal semantics. We
    formalised direct event emissions and recursive subagent invocations and proved that the
    recursive closure \(\E_\infty\) of these operations from single causal roots in the ambient
    space \(\G\) of finite sequences of typed temporal graphs consists entirely of finite sequences
    of CTEGs. We further observed that \(\E_\infty\) is the least fixed point of a monotone
    operator \(\varphi\) on \(\pwset(\G)\), the power set of \(\G\), obtained as the limit of the
    ascending \(\varphi\)-chain \(\emptyset \subseteq \E_0 \subseteq \E_1 \subseteq \cdots\).
    Although the full hierarchy arises naturally, stabilisation occurs already at \(\E_1\). This
    reflects the opacity of the subagent invocation interface, wherein only the final subagent
    execution structure \(H_m\) is exposed to the parent, while its construction history \((H_i)\)
    is discarded. Every CTEG admits an \(\E_0\)-construction by topological sorting, hence \(H_m\)
    might as well have come from an \(\E_0\)-invocation from the parent's point of view.

    The central practical consequence is the identification of a compositional mechanism by which
    well-formed global causal structure may be assembled from delegated local subagent trace
    constructions without centralised coordination. The execution model is append-only by design:
    No operation retracts or modifies existing nodes, preserving the recorded causal structure. The
    arborescent structure of CTEGs enables natural relational database encodings of agent session
    histories as well as tamper-evident cryptographic commitments via Merkle tree hashes.
    Interrupted sessions preserve these properties and therefore remain fully serialisable under
    the same encoding.

    Capturing recursion depth as a structurally observable global execution trace invariant is an
    interesting direction for future work and can likely be achieved by carefully preserving more
    of the subagent execution history at the invocation interfaces or by treating subagent
    invocations as atomic nodes in the parent agents' execution graphs. The latter would
    entail global execution records no longer be fully resolved in place, instead giving rise to
    trees of local views. We have also left open the investigation of direct agent-to-agent
    communications outside strict parent-child relationships, as well as the paradigm of more
    general multi-parenthood causal semantics.

    We have distinguished between well-formed and correct causal structure, arguing that the latter
    must be assessed relative to a choice of local parenthood semantics, the selection of which is
    orthogonal to the CTEG formalism as the encoding layer and outside the scope of the current
    paper. A natural extension of this work would be to formalise such a notion of local semantics,
    constraining the admissible parenthood attributions available at each step in accordance with
    well-founded causal principles. From this perspective, the theory developed herein describes
    the \emph{maximal} case, in which every generically available local parenthood attribution is
    allowed.

    \printbibliography
\end{document}